\documentclass[12 pt]{extarticle}
\usepackage{amsmath}
\usepackage{amssymb}
\usepackage[a4paper, total={6in, 9in}]{geometry}
\usepackage{amsthm,url,tikz}
\usepackage[utf8]{inputenc}
\usepackage[english]{babel}
\usepackage{graphicx}

\usepackage{makecell}

\usepackage[pagewise]{lineno}%\linenumbers

\usepackage{cases} 
\usepackage{enumerate}
\usepackage{hyperref}
\usepackage{adjustbox}
\numberwithin{equation}{section}
\usepackage{makecell}
\usepackage{bm}
\usepackage{xcolor}

\newtheorem{theorem}{Theorem}[section]

\newtheorem{corollary}[theorem]{Corollary}
\newtheorem{lemma}[theorem]{Lemma}

\theoremstyle{definition} 
 
\newtheorem{remark}{Remark}[section] 
\newtheorem{example}{Example}[section]

\usepackage[nottoc,notlot,notlof]{tocbibind} 
\title{Linear Complementary Pairs of Quasi-Cyclic and Quasi-Twisted Codes}
\author{Kanat Abdukhalikov \\	
Department of Mathematical Sciences, \\
UAE University, PO Box 15551, Al Ain, UAE\\
Email: abdukhalik@uaeu.ac.ae \bigskip  \\  
Duy Ho \\
Department of Mathematics and Statistics, \\
UiT The Arctic University of Norway, Tromsø
9037, Norway\\
Email: duyho92@gmail.com \bigskip \\
San Ling \\	
School of Physical and Mathematical Sciences, \\
Nanyang Technological University, Singapore 637371\\ 
 Email: lingsan@ntu.edu.sg \bigskip \\
Gyanendra K. Verma \\	
Department of Mathematical Sciences, \\
UAE University, PO Box 15551, Al Ain, UAE\\
Email:  gkvermaiitdmaths@gmail.com}
\date{}

\begin{document}
	\maketitle
\begin{abstract} 
 In this paper, we provide a polynomial characterization of linear complementary pairs of quasi-cyclic and quasi-twisted codes of index $2$. We also give several examples of linear complementary pairs of quasi-cyclic and quasi-twisted codes with optimal security parameters.
\end{abstract}

\section{Introduction}
Quasi-cyclic codes are a class of codes that includes cyclic codes.  Townsend and  Weldon introduced these codes \cite{Townsend1967} in 1967. It is already known that quasi-cyclic codes are asymptotically good. Beyond classical communication and data storage applications, quasi-cyclic codes are used in modern post-quantum cryptographic systems, such as QC-MDPC code-based protocols, due to their balance of security and compact key sizes. A systematic study of the algebraic structure of quasi-cyclic codes over finite fields and rings has been done in \cite{Ling2006,Ling2001,Ling2003,Ling2005}.  In 2023, Abdukhalikov et al. \cite{Abdukhalikov2023} investigated one-generator quasi-cyclic codes over finite fields and determined the duals in polynomial representation. Lally and Fitzpatrick \cite{lally2001} gave the polynomial characterization of quasi-cyclic codes by proving that every
 quasi-cyclic code has a generating set consisting of $\ell$-tuples polynomials with certain properties derived from reduced Gr\"obner
 basis. Based on these structural properties, more asymptotic results, minimum distance bounds, and further applications of quasi-cyclic codes were obtained in the literature (\cite{Luo2023,Semenov2012,Zeh2015}). Further, using polynomial characterization, quasi-cyclic codes of index $2$ have been investigated in \cite{Abdukhalikov20251}. Quasi-twisted codes are another generalization of cyclic codes that includes constacyclic and quasi-cyclic codes. Additionally, they are known to be asymptotically good (\cite{chepyzhov1993,wu2020,kasami1974}). Many record-breaking codes with excellent parameters have been constructed within this class (e.g., see \cite{ackerman2011,aydin2017,qian2019}).

A pair of linear codes of the same length over a finite field is called a linear complementary pair (LCP) of codes if the codes intersect trivially and the direct sum of the codes is the whole ambient space. LCPs of codes have been extensively studied due to their nice algebraic structure and wide applications in cryptography. Ngo et al. \cite{Ngo2015} introduced LCPs of codes in 2015 and showed their applications in resistance against physical attacks like side-channel attacks (SCA) and fault injection attacks (FIA). This method is known as Direct Sum Masking (DSM). The security parameters of a given LCP  $(C,D)$ over a finite field depend on $\min \{d(C),d(D^\perp)\}$, where $d(C)$ denotes the minimum Hamming distance of  $C$ and $d(D^\perp)$ denotes the dual distance of $D$ (for details, see \cite{Carlet2016,carlet2018}). Linear complementary dual (LCD) codes are a special case of an LCP of codes where $D=C^{\perp}$. The security parameters of an LCD code depend only on the minimum distance of the code. 

In 2018, Carlet et al. \cite{carlet2018}
 characterized LCPs of quasi-cyclic codes using the  Chinese Remainder Theorem (CRT) decomposition of codes and also provided necessary and sufficient conditions for pairs of constacyclic and double circulant codes to be an LCP of codes. Moreover, they proved that if $(C, D)$ is an LCP of constacyclic or 2D cyclic code, then $C$ and $D^{\perp}$ are permutation equivalent. G\"uneri et al. \cite{Guneri2018} showed the equivalence of $C$ and $D^{\perp}$ holds for an LCP of  $mD$-cyclic codes as well. These results show that finding LCD codes and LCPs of codes is the same in the context of security parameters. However, $C$ and $D^\perp$ need not be permutation equivalent for an LCP of quasi-cyclic codes \cite[Example 3.3]{carlet2018}.  In \cite{carlet2019}, Carlet et al. proved that for $q>2$, there exists an LCP of linear codes with security parameters as good as the optimal distance of linear codes, and for $q=2$, the security parameters of optimal binary LCPs of linear codes are lower bounded by one less than the optimal distance of linear codes. In \cite{Choi2022,Guneri2023,Li2023}, the authors discussed optimal binary LCPs of linear codes. In this work, we give a polynomial characterization of  LCPs of quasi-cyclic and quasi-twisted codes with index $2$. Using these characterizations, we provide several examples of optimal LCPs of quasi-twisted codes over $\mathbb{F}_q$. The rest of the paper is organized as follows.

 In Section \ref{pre}, we recall the basics of LCPs of codes and the decomposition of quasi-cyclic codes. In Section \ref{secgen2lcp}, we establish necessary and sufficient conditions for LCPs of quasi-cyclic codes of index $2$. We discuss LCPs of one-generator quasi-cyclic codes in Section \ref{secgen1lcp}. We characterize LCPs of quasi-twisted codes in Section \ref{secqtlcp}. In Section \ref{secexample}, we give several examples of LCPs of quasi-cyclic codes.

\section{Preliminaries}\label{pre}
We denote $\mathbb{F}_q$ as a finite field with $q$ elements. A linear code over $\mathbb{F}_q$ of length $n$  is a subspace of $\mathbb{F}_q^n$. We say $C$ is an $[n,k]$ code over $\mathbb{F}_q$ code if $C$ is a linear code of length $n$ and $\dim_{\mathbb{F}_q}(C)=k$. Let $C$ and $D$ be linear codes of length $n$ over $\mathbb{F}_q$. Then $(C, D)$ is called a linear complementary pair (LCP) of codes if they have the trivial intersection and their direct sum is the ambient space $\mathbb{F}_q^n$. That is,
$$ C\oplus D=\mathbb{F}_q^n.$$
Note that if $C$ has parameters $[n,k]$, then $D$ must have parameters $[n,n-k]$.   Linear complementary dual (LCD) codes are a special class of LCPs of codes where $D$ is the dual code of $C$. Equivalently, $C$  is an LCD code if $C\cap C^{\perp}=\{0\}$. 

 Let $\lambda \in \mathbb{F}_q^*$. We define the cyclic $\lambda$-shift operator as 
 $$T_{\lambda}(x_0,x_1,\dots,x_{n-1})=(\lambda x_{n-1},x_0,\dots,x_{n-2}).$$
 A linear code is said to be a $\lambda$-quasi-twisted code of length $n$ and index $\ell$ (in short, $(\lambda,\ell)$-quasi-twisted) if $T_{\lambda}^\ell(c)\in C$ for all $c\in C$. If $\lambda=1$, then $C$ is  called a quasi-cyclic code of index $\ell$. When $\ell=1$, then $C$ is called  a $\lambda$-constacyclic code, and when $\ell=1$ $\lambda=1$, $C$ is called a cyclic code. If $C$ is a $\lambda$-quasi-twisted code of length $n$ and index $\ell$, then one can assume that $\ell$ divides $n$. Thus, we assume that $n=m\ell$. 

 Let $R_\lambda=\mathbb{F}_q[x]/\langle x^m-\lambda\rangle$ and for $\lambda=1$, we denote $R_1$ by $R$. Note that $\lambda$-constacyclic codes over $\mathbb{F}_q$ of length $m$ are ideals of $R_\lambda$. Let $C$ be a $\lambda$-quasi-twisted code of length $n=m\ell$ and index $\ell$ over $\mathbb{F}_q$. It is well known that $C$ can be identified as an $R_\lambda$-submodule of $R_\lambda^\ell$ (for instance, see \cite{jia2012}). The correspondence is given by
 $$(c_{0,0},c_{0,1},\dots,c_{0,\ell-1},c_{1,0},\dotsc_{1,\ell-1},\dots,c_{m-1,0},\dots,c_{m-1,\ell-1})\mapsto (c_0(x),c_1(x),\dots c_\ell(x))\in R_\lambda^\ell,$$
 where $c_i(x)=c_{0,i}+c_{1,i}x+\dots+c_{m-1,i}x^{m-1}\in R_\lambda$. 

Throughout the article, we assume $\gcd(m,q)=1$.  We decompose quasi-cyclic codes ($\lambda=1$) into constituent codes (for details, see \cite[Chapter 7, Quasi-cyclic codes p.132]{Huffman2021}). Let $x^m-1=f_1(x)f_2(x)\dots f_t(x)$ be the factorization into a product of distinct irreducible polynomials over $\mathbb{F}_q$. For each $i$,  let $F_i:=\mathbb{F}_q[x]/\langle f_i(x) \rangle$. 
Let $\xi$ be a primitive $m^{\text{th}}$ root of unity in some extension of $\mathbb{F}_q$. Let $\xi^{u_i}$ be a root of $f_i(x)$.
Then $F_i= \mathbb{F}_q(\xi^{u_i})$. By the Chinese
 Remainder Theorem (CRT), $R=\mathbb{F}_q[x]/\langle x^m-1\rangle$ can be decomposed as
 $$
R \cong \bigoplus_{i=1}^{t} F_i.
 $$ 
The isomorphism between $R$ and its CRT decomposition is given by 
\[
a(x) \mapsto \left(a(\xi^{u_1}), \dots, a(\xi^{u_t})\right). 
\]
This isomorphism extends naturally to  $R^\ell$, which implies that
 $$
R^\ell \cong \bigoplus_{i=1}^{t} F_i^\ell.
 $$ 
Then, a QC code $C$ of index $\ell$ can be  decomposed as 
\begin{equation} \label{crtC}
C \cong \bigoplus_{i=1}^{t} C_i, 
\end{equation} 
where each component code is a linear code of length $\ell$ over the base field $F_i$.  The component codes  $C_i$ are called the \textit{constituents} of $C$.

The constituents can be described in terms of the generators of $C$. Namely, if $ C$ is an $r$-generator QC code with generators
\[
 \{ (a_{1,1}(x),\dots ,a_{1,\ell}(x)) ,\dots, (a_{r,1}(x),\dots,a_{r,\ell} (x) ) \} \subset R^\ell,
\]
then 
\begin{align*}
C_i &= \text{Span}_{F_i} \{ (a_{b,1}(\xi^{u_i}),\dots ,a_{b,\ell}(\xi^{u_i})) : 1 \le b \le r \}, \text{ for } 1 \le i \le t.
\end{align*}

We have the following characterization of  QC LCP of codes from \cite[Theorem 3.1]{carlet2018}.

\begin{theorem} \label{qclcp} Let $C$ and $D$ be  $q$-ary QC codes of length $m\ell$ and index $\ell$ with  CRT decompositions  $C \cong \bigoplus_{i=1}^{t} C_i$  and $D \cong \bigoplus_{i=1}^{t} D_i$, respectively.
Then $(C,D)$ is an LCP of codes if and only if 
$(C_i, D_i)$ is an LCP of codes for all $1\leq i\leq t$. 
 \end{theorem}

\section{LCPs of quasi-cyclic codes of index $2$}\label{secgen2lcp}
In \cite{lally2001,Abdukhalikov20251}, Lally and Fitzpatrick showed that a quasi-cyclic code of index $\ell$ can be generated by the rows of an upper triangular $\ell \times \ell$  polynomial matrix satisfying certain conditions. We state for $\ell=2$.
\begin{theorem} \label{qc2}
Let $C$ be a quasi-cyclic code of length $2m$ and index $2.$ 
Then $C$ is generated by two elements $(g_{11}(x), g_{12}(x))$ and $(0, g_{22}(x))$ such that they satisfy the following conditions: 
\begin{gather}  
g_{11}(x) \mid (x^m-1) \text{ and } g_{22}(x) \mid(x^m-1), \notag \\
\deg g_{12}(x) < \deg g_{22}(x), \tag{$\ast$}\\
g_{11}(x)g_{22}(x) \mid (x^m-1)g_{12}(x). \notag
\end{gather}
 Moreover, in this case $\dim (C) = 2m-\deg g_{11}(x)-\deg g_{22}(x)$.
\end{theorem}

\begin{remark} \label{remarkgcd}   If $\gcd(q,m)=1$, then the condition 
$$g_{11}(x)g_{22}(x) \mid (x^m-1)g_{12}(x)$$ 
in Theorem \ref{qc2} is equivalent to the condition $\gcd(g_{11}(x),g_{22}(x)) \mid g_{12}(x)$, since $x^m-1$ has no multiple roots.
\end{remark}

Let $C$ and $D$ be two quasi-cyclic codes of length $2m$ and index $2$, generated by the rows of the matrices $G$ and $H$, respectively, where 
\begin{equation*}
    G= \begin{bmatrix}
    g_{11}(x)&g_{12}(x)\\
    0& g_{22}(x)
\end{bmatrix} \text{ and } H= \begin{bmatrix}
    f_{11}(x)&f_{12}(x)\\
    0& f_{22}(x)
\end{bmatrix}
\end{equation*}
and $f_{ij}(x)$, $g_{ij}(x)$ satisfy Condition $(\ast)$.

Then the constituent codes $C_i$ and $D_i$ are generated by the rows of the matrices
$$
G_i=  \begin{bmatrix}
g_{11}(\xi^{u_i}) & g_{12}(\xi^{u_i}) \\
0 & g_{22}(\xi^{u_i}) 
\end{bmatrix} \text{ and }
H_i=  \begin{bmatrix}
f_{11}(\xi^{u_i}) & f_{12}(\xi^{u_i}) \\
0 & f_{22}(\xi^{u_i}) 
\end{bmatrix},$$
respectively. 

Let $g(x)=\gcd(g_{11}(x),g_{22}(x))$ and $f(x)=\gcd(f_{11}(x),f_{22}(x))$. Then, by Remark \ref{remarkgcd}, $g(x)\mid g_{12}(x)$ and $f(x)\mid f_{12}(x)$.
% Let  $l(x)= (x^m-1)/\text{lcm}(g_{11}(x),g_{22}(x))$ and $h(x)= (x^m-1)/\text{lcm}(f_{11}(x),f_{22}(x))$. 
Let $g_{11}(x)=g(x)   g_{11}'(x), g_{22}(x)=g(x) g_{22}'(x)$, and $f_{11}(x)=f(x)   f_{11}'(x), f_{22}(x)=f(x) f_{22}'(x)$. 

Now, we give a new characterization for a pair of quasi-cyclic codes of index $2$ to be an LCP of codes. We use the notations as above.

\begin{theorem}\label{qslcp}
Let $C$ and $D$ be two quasi-cyclic codes of length $2m$ and index $2$ generated by $(g_{11}(x),g_{12}(x))$, $(0,g_{22}(x))$ and $(f_{11}(x),f_{12}(x))$, $(0,f_{22}(x))$, respectively, satisfying Condition $(\ast)$. Then $(C,D)$ is an LCP of codes if and only if all of the following conditions are true:
	\begin{enumerate}[(I)]
       \item  $\gcd(f_{11}(x),g_{11}(x))=1$.
		\item  $g_{11}(x)g_{22}(x)f_{11}(x)f_{22}(x)=(x^m-1)^2$.
		\item $\gcd(g_{22}'(x),f_{22}'(x), g_{11}(x)f_{12}(x)-g_{12}(x)f_{11}(x))=1$.
	\end{enumerate}
\end{theorem}
We provide the proof of Theorem \ref{qslcp} through the following sequence of lemmata.
\begin{lemma} \label{nI} If $(C,D)$ is an LCP of codes, then  (I) holds. 
\end{lemma}
\begin{proof}
     Assume that $\gcd(f_{11}(x),g_{11}(x))\neq 1$. This implies that there exists $f_i(x)$ such that $f_i(x)\mid g_{11}(x)$ and $f_i(x)\mid f_{11}(x)$.

Then $g_{11}(\xi^{u_i})= f_{11}(\xi^{u_i})=0$, and $C_i$ and $D_i$ are generated by the rows of the  matrices
$$
G_i=\begin{bmatrix}
g_{11}(\xi^{u_i}) & g_{12}(\xi^{u_i}) \\
0 & g_{22}(\xi^{u_i}) 
\end{bmatrix}    = \begin{bmatrix}
0 & g_{12}(\xi^{u_i}) \\
0 & g_{22}(\xi^{u_i}) 
\end{bmatrix} \text{ and } H_i=\begin{bmatrix}
f_{11}(\xi^{u_i}) & f_{12}(\xi^{u_i}) \\
0 & f_{22}(\xi^{u_i}) 
\end{bmatrix}    = \begin{bmatrix}
0 & f_{12}(\xi^{u_i}) \\
0 & f_{22}(\xi^{u_i})  
\end{bmatrix},
$$ respectively.
This implies that  $\dim(C_i+D_i)\neq 2$.  Therefore, $(C_i,D_i)$ is not an LCP of codes.     
\end{proof}

\begin{lemma} \label{nII} If $(C,D)$ is an LCP of codes, then  (II) holds. 
\end{lemma}
\begin{proof}
 Let $(C,D)$ be an LCP of codes. By Lemma \ref{nI}, $(I)$ holds. We claim that
\begin{align}\label{ii.1}
g(x)= \frac{x^m-1}{\text{lcm}(f_{11}(x),f_{22}(x))}, \text{ where } g(x)=\gcd(g_{11}(x),g_{22}(x)).
\end{align}
On the contrary, assume $g(x)\neq \frac{x^m-1}{\text{lcm}(f_{11}(x),f_{22}(x))}$. This implies there exists irreducible $f_i(x)$ such that:\\
    $(i)$ $f_i(x)\mid g(x)$ and $f_i(x)\nmid \frac{x^m-1}{\text{lcm}(f_{11}(x),f_{22}(x))} $ or \\
    $(ii)$ $f_i(x)\nmid g(x)$ and $f_i(x)\mid \frac{x^m-1}{\text{lcm}(f_{11}(x),f_{22}(x))} $.\\\\
    $(i)$ If $f_i(x)\mid g(x)$ and $f_i(x)\nmid \frac{x^m-1}{\text{lcm}(f_{11}(x),f_{22}(x))} $, then $f_i(x)\mid g_{11}(x),g_{22}(x),g_{12}(x)$ and $f_i(x)\mid  \text{lcm}(f_{11}(x),f_{22}(x))$. This implies that $f_i(x)\mid f_{22}(x)$ since $\gcd(g_{11}(x),f_{11}(x))=1$.  Thus, $g_{11}(\xi^{u_i})= g_{22}(\xi^{u_i})=g_{12}(\xi^{u_i})=f_{22}(\xi^{u_i})=0$ and $C_i$ and $D_i$ are generated by the rows of the  matrices,
$$
G_i=\begin{bmatrix}
g_{11}(\xi^{u_i}) & g_{12}(\xi^{u_i}) \\
0 & g_{22}(\xi^{u_i}) 
\end{bmatrix}    = \begin{bmatrix}
0 & 0 \\
0 & 0 
\end{bmatrix} \text{ and } H_i=\begin{bmatrix}
f_{11}(\xi^{u_i}) & f_{12}(\xi^{u_i}) \\
0 & f_{22}(\xi^{u_i}) 
\end{bmatrix}    = \begin{bmatrix}
f_{11}(\xi^{u_i}) & f_{12}(\xi^{u_i}) \\
0 & 0  
\end{bmatrix},
$$ respectively.
This implies that  $\dim(C_i)+\dim(D_i)\neq 2$.  Therefore, $(C_i, D_i)$ is not an LCP of codes.  \\
$(ii)$ If $f_i(x)\nmid g(x)$ and $f_i(x)\mid \frac{x^m-1}{\text{lcm}(f_{11}(x),f_{22}(x))}$, then $f_i(x)\nmid g_{11}(x)$ or $f_i(x)\nmid g_{22}(x)$, and  $f_i(x)\nmid \text{lcm}(f_{11}(x),f_{22}(x))$   (equivalently, $f_i(x)\nmid f_{11}(x),f_{22}(x)$). This implies that 
 $g_{11}(\xi^{u_i})\neq 0$ or $g_{22}(\xi^{u_i})\neq 0$, and  $f_{11}(\xi^{u_i})\neq 0 \neq f_{22}(\xi^{u_i})$. Then $C_i$ and $D_i$ are generated by the rows of the matrix 
 $$
G_i=\begin{bmatrix}
g_{11}(\xi^{u_i}) & g_{12}(\xi^{u_i}) \\
0 & g_{22}(\xi^{u_i}) 
\end{bmatrix} \text{ and } H_i=\begin{bmatrix}
f_{11}(\xi^{u_i}) & f_{12}(\xi^{u_i}) \\
0 & f_{22}(\xi^{u_i}) 
\end{bmatrix},$$ respectively. This implies that $\dim(C_i)\geq 1$ and $\dim(D_i)=2$. Thus $C_i\cap D_i\neq \{0\}$. Therefore, $(C_i, D_i)$ is not an LCP of codes, a contradiction. Thus, Eq. \ref{ii.1} holds.\\
Using similar arguments, we get 
\begin{align}\label{ii.2}
    f(x)= \frac{x^m-1}{\text{lcm}(g_{11}(x),g_{22}(x))}, \text{ where } f(x)=\gcd(f_{11}(x),f_{22}(x)).
\end{align}
By Eqs. \ref{ii.1} and \ref{ii.2}, we have 
\begin{align*}
    &g(x)\text{lcm}(f_{11}(x),f_{22}(x))\cdot f(x)\text{lcm}(g_{11}(x),g_{22}(x))=(x^m-1)^2,\\
    i.e.,\ \  &g_{11}(x)g_{22}(x)f_{11}(x)f_{22}(x)=(x^m-1)^2.
    \end{align*}  
    This completes the proof.
\end{proof}

\begin{lemma}\label{nIII}
    If $(C,D)$ is an LCP of codes, then $(III)$ holds.
\end{lemma}

\begin{proof}
 Assume $\gcd(g_{22}'(x),f_{22}'(x), g_{11}(x)f_{12}(x)-g_{12}(x)f_{11}(x))\neq 1$. Then there exists $f_i(x)$ such that $f_i(x)\mid g_{22}'(x), f_{22}'(x)$ and $f_i(x)\mid (g_{11}(x)f_{12}(x)-g_{12}(x)f_{11}(x)) $. This implies that $g_{22}'(\xi^{u_i})=f_{22}'(\xi^{u_i})=g_{11}(\xi^{u_i})f_{12}(\xi^{u_i})-g_{12}(\xi^{u_i})f_{11}(\xi^{u_i})=0$ and $C_i$ and $D_i$ are generated by the rows of the matrices 
 $$
G_i=\begin{bmatrix}
g_{11}(\xi^{u_i}) & g_{12}(\xi^{u_i}) \\
0 & g_{22}(\xi^{u_i}) 
\end{bmatrix}    = \begin{bmatrix}
g_{11}(\xi^{u_i}) & g_{12}(\xi^{u_i}) \\
0 & 0 
\end{bmatrix}$$
$$\text{ and } H_i=\begin{bmatrix}
f_{11}(\xi^{u_i}) & f_{12}(\xi^{u_i}) \\
0 & f_{22}(\xi^{u_i}) 
\end{bmatrix}    = \begin{bmatrix}
f_{11}(\xi^{u_i}) & f_{12}(\xi^{u_i}) \\
0 & 0  
\end{bmatrix},
$$ respectively. Then $C_i\cap D_i=\{0\}$ if and only if the determinant of the matrix $ A=\begin{bmatrix}
 g_{11}(\xi^{u_i}) & g_{12}(\xi^{u_i})\\ 
 f_{11}(\xi^{u_i}) & f_{12}(\xi^{u_i})
\end{bmatrix}$ is non-zero, but $g_{11}(\xi^{u_i})f_{12}(\xi^{u_i})-g_{12}(\xi^{u_i})f_{11}(\xi^{u_i})=0$ implies that $det(A)=0$. Thus $C_i\cap D_i\neq \{0\}$. Therefore, $(C_i,D_i)$ is not an LCP of codes.
\end{proof}

\begin{lemma}\label{converse}
    If $(I),(II),(III)$ hold, then $(C,D)$ is an LCP of codes, that is $C_i\cap D_i=\{0\}$ and $\dim(C_i)+\dim(D_i)=2$ for all $1\leq i\leq t$.
\end{lemma}
\begin{proof}
By $(II)$, we have 
\begin{align}\label{con.1}
    g(x) g(x)g_{11}'(x)g_{22}'(x)f_{11}'(x)f_{22}'(x)f(x)f(x)=(x^m-1)^2,
\end{align}
where $g(x)=\gcd(g_{11}(x),g_{22}(x))$ and $f(x)=\gcd(f_{11}(x),f_{22}(x))$.
For each $i$, the irreducible polynomial $f_i$ divides at least one of the factors $g(x)$, $g_{11}'(x)$, $g_{22}'(x)$, $f(x),f_{11}'(x),f_{22}'(x)$ of $x^m-1$, which leads to the following six cases.  
\\
$\bm{(i)}$ Let $f_i(x)\mid g(x)$. Then $f_i(x)\mid g_{11}(x),g_{22}(x),g_{12}(x)$. Also, $f_i^3(x)\nmid (x^m-1)^2$ implies that $f_i(x)\nmid f_{11}'(x),f_{22}'(x),f(x)$. It follows that  
$$
G_i=   \begin{bmatrix}
0 & 0 \\
0 & 0 
\end{bmatrix} , \text{ and }
H_i=   \begin{bmatrix}
f_{11}(\xi^{u_i}) & f_{12}(\xi^{u_i}) \\
0 & f_{22}(\xi^{u_i})  
\end{bmatrix},$$
where $f_{11}(\xi^{u_i})\neq0\neq f_{22}(\xi^{u_i})$.
It can then be readily checked that $C_i \cap D_i= \{0\}$ and $\dim(C_i)+\dim(D_i)=2$.\\
$\bm{(ii)}$  Let $f_i(x)\mid g_{11}'(x)$. Then $f_i(x)\nmid g_{22}'(x)$. Also, $f_i^3(x)\nmid (x^m-1)^2$ implies that  $f_i(x)\nmid g(x),f(x)$.  Hence, either $f_i(x)\mid f_{11}'(x)$ or $f_i(x)\mid f_{22}'(x)$. By $(I)$, $f_i(x)\nmid f_{11}(x)$ and thus $f_i(x)\mid f_{22}'(x)$. It follows that 
$$
G_i=   \begin{bmatrix}
0 & g_{12}(\xi^{u_i}) \\
0 &  g_{22}(\xi^{u_i}) 
\end{bmatrix}, \text{ and }
H_i=   \begin{bmatrix}
f_{11}(\xi^{u_i}) & f_{12}(\xi^{u_i}) \\
0 & 0  
\end{bmatrix},
$$ where $g_{22}(\xi^{u_i})\neq0\neq f_{11}(\xi^{u_i})$.
It can then be readily checked that $C_i \cap D_i= \{0\}$ and $\dim(C_i)+\dim(D_i)=2$.\\
$\bm{(iii)}$ Let $f_i(x)\mid g_{22}'(x)$.  Then $f_i(x)\nmid g(x),g_{11}'(x),f(x)$. Thus either $f_i(x)\mid f_{11}'(x)$ or $f_i(x)\mid f_{22}'(x)$. \\
1. If $f_i(x)\mid f_{11}'(x)$ then $f_i(x)\nmid f_{22}'(x)$ and thus
$$
G_i=   \begin{bmatrix}
g_{11}(\xi^{u_i}) & g_{12}(\xi^{u_i}) \\
0 &  0
\end{bmatrix} 
 \text{ and }
H_i=   \begin{bmatrix}
0 & f_{12}(\xi^{u_i}) \\
0 &   f_{22}(\xi^{u_i})
\end{bmatrix},
$$ where $g_{11}(\xi^{u_i})\neq 0\neq f_{22}(\xi^{u_i})$.
It can then be readily checked that $C_i \cap D_i= \{0\}$ and $\dim(C_i)+\dim(D_i)=2$.\\
2. If $f_i(x)\mid f_{22}'(x)$ then $f_i(x)\nmid f_{11}'(x)$ and thus
$$
G_i=   \begin{bmatrix}
g_{11}(\xi^{u_i}) & g_{12}(\xi^{u_i}) \\
0 &  0
\end{bmatrix} \text{ and }
H_i=   \begin{bmatrix}
 f_{11}(\xi^{u_i})& f_{12}(\xi^{u_i}) \\
0 &   0
\end{bmatrix}.
$$
By $(III)$, $f_i(x)\nmid (g_{11}(x)f_{12}(x)-g_{12}(x)f_{11}(x)) $ (since $f_i(x)\mid g_{22}'(x),f_{22}'(x)$). This implies that $g_{11}(\xi^{u_i})f_{12}(\xi^{u_i})-g_{12}(\xi^{u_i})f_{11}(\xi^{u_i})\neq0$. Therefore,  $C_i \cap D_i= \{0\}$ and $\dim(C_i)+\dim(D_i)=2$.\\
$\bm{(iv)}$ Let $f_i(x)\mid f(x)$. The proof is similar to the case $\bm{(i)}$ by interchanging 
%the generating polynomials corresponding to 
codes $C_i$ and $D_i$.\\
% Then  $f_i(x)\nmid g(x),g_{11}'(x),g_{22}'(x)$ and $f_i(x)\mid f_{11}(x),f_{22}(x),f_{12}(x)$.} Thus 
% $$
% G_i=   \begin{bmatrix}
% g_{11}(\xi^{u_i}) & g_{12}(\xi^{u_i}) \\
% 0 & g_{22}(\xi^{u_i}) 
% \end{bmatrix}, \text{ and }
% H_i=   \begin{bmatrix}
% 0& 0 \\
% 0 & 0 
% \end{bmatrix},
% $$ where $g_{11}(\xi^{u_i})\neq 0\neq g_{22}(\xi^{u_i})$.
% It can then be readily checked that $C_i \cap D_i= \{0\}$ and $\dim(C_i)+\dim(D_i)=2$.\\
$\bm{(v)}$ Let $f_i(x)\mid f_{11}'(x)$. The proof is similar to the case $\bm{(ii)}$ by interchanging 
%the generating polynomials corresponding to 
codes $C_i$ and $D_i$.\\
$\bm{(vi)}$ Let $f_i(x)\mid f_{22}'(x)$. The proof is similar to the case $\bm{(iii)}$ by interchanging 
%the generating polynomials corresponding to 
codes $C_i$ and $D_i$.
\end{proof}

\begin{proof}[Proof of Theorem \ref{qslcp}]
 We recall that we want to prove that $(C,D)$ is an LCP of codes if and only if the conditions (I), (II), and  (III) hold. The ``if" direction follows from Lemma \ref{converse}. The ``only if" direction follows from Lemmata \ref{nI}, \ref{nII}, and \ref{nIII}.
    
\end{proof}

\section{ LCPs of one-generator quasi-cyclic codes  of index 2}\label{secgen1lcp}

In this section, we consider one-generator quasi-cyclic codes of index $2$.
% Suppose $C$ is generated by one element 
%  $(g_{11}(x), g_{12}(x))$, where $g_{11}(x) \mid (x^m-1)$.  Let $g(x)=\gcd(g_{11}(x),g_{12}(x))$,
% $g_{11}(x)=g(x)g_{11}'(x)$, $g_{12}(x)=g(x)g_{12}'(x)$. Let 
% \[
% g_{22}(x) = \dfrac{x^m-1}{g_{11}'(x)}.
% \]
The following lemma from \cite[Lemma 2.4]{Abdukhalikov20252} shows how one-generator quasi-cyclic codes are generated by two generators satisfying Condition $(\ast)$.

\begin{lemma} \label{lemma1gen} 
Let $\gcd(q,m)=1$ and let $C$ be a quasi-cyclic code generated by one element $(g_{11}(x),g_{12}(x))$, where $g_{11}(x) \mid (x^m-1)$.  
Let $g(x)=\gcd(g_{11}(x),g_{12}(x))$, $g_{11}(x)=g(x)g_{11}'(x)$, $g_{12}(x)=g(x)g_{12}'(x)$. 
Let 
\[
g_{22}(x) = \dfrac{x^m-1}{g_{11}'(x)}.
\]
% Let $\tilde{g}_{12}(x)=g_{12}(x) \mod{g_{22}(x)}$.  
Then the code $C$ is generated by two elements $(g_{11}(x),g_{12}(x) \pmod{g_{22}(x)})$ and $(0,g_{22}(x))$ satisfying Condition $(\ast)$ with
 $\gcd(g_{11}(x),g_{22}(x))=g(x)$.
\end{lemma}

From the above Lemma \ref{lemma1gen}, if $C$ is one-generator quasi-cyclic code of index $2$, then we can assume that $C$ is   generated by two elements $(g_{11}(x), \tilde{g}_{12}(x))$ and $(0, g_{22}(x))$ satisfying Condition $(\ast)$, where 
  $\tilde{g}_{12}(x)=g_{12}(x) \pmod{g_{22}(x)}$. Furthermore, also by Lemma \ref{lemma1gen} we have that  $\gcd(g_{11}(x),g_{22}(x))=g(x)$.  In addition,  $(0,g_{22}(x))\in C$ (see proof of \cite[Lemma 2.4]{Abdukhalikov20252}). It follows that  code $C$ generated by the element $(g_{11}(x),g_{12}(x))$ is the same as the code generated by the element $(g_{11}(x),\tilde{g}_{12}(x))$. Thus, without loss of generality, from now on, we will write $g_{12}(x)$ instead of $\tilde{g}_{12}(x)$.
  We have 
  \begin{equation}\label{1geneq}
      \text{lcm}(g_{11}(x),g_{22}(x)) = \dfrac{g_{11}(x)\cdot g_{22}(x)}{\gcd(g_{11}(x),g_{22}(x))} = \dfrac{g(x) \cdot g_{11}'(x) \cdot (x^m-1)}{g(x) \cdot g_{11}'(x)} = x^m-1.
  \end{equation}
% \begin{lemma}
%     Let $g(x)=1$. Then the following are equivalent
%     \begin{enumerate}
%     \item $\gcd(f_{11}(x),g_{11}(x))=1$
%         \item $\gcd(g_{11}(x),g_{11}(x)f_{12}(x)-g_{12}(x)f_{11}(x))=1$
%     \end{enumerate}
% \end{lemma}

% \begin{lemma}
%     Let $f(x)=1$. Then the following are equivalent
%     \begin{enumerate}
%     \item $\gcd(f_{11}(x),g_{11}(x))=1$
%         \item $\gcd(f_{11}(x),g_{11}(x)f_{12}(x)-g_{12}(x)f_{11}(x))=1$
%     \end{enumerate}
% \end{lemma}

\begin{lemma}\label{gen1lemma}
    Let $g_{11}(x)$, $g_{12}(x)$, $f_{11}(x)$, and $f_{12}(x)$ in $\mathbb{F}_q[x]$ be such that  $g_{11}(x)$ and $f_{11}(x)$ divide $x^m-1$. Let $\gcd(g_{11}(x),g_{12}(x))=1$, $\gcd(f_{11}(x),f_{12}(x))=1$ and $\gcd(f_{11}(x),g_{11}(x))=1$. Then the following are equivalent.
    \begin{enumerate}[(1)]
        \item $\gcd\left (\frac{x^m-1}{f_{11}(x)},\frac{x^m-1}{g_{11}(x)},g_{11}(x)f_{12}(x)-g_{12}(x)f_{11}(x)\right )=1$.
        \item $\gcd(x^m-1,g_{11}(x)f_{12}(x)-g_{12}(x)f_{11}(x) )=1$.
    \end{enumerate}
\end{lemma}
\begin{proof}
    $(2)$ implies $(1)$ trivially. Let $\gcd(x^m-1,g_{11}(x)f_{12}(x)-g_{12}(x)f_{11}(x) )\neq 1$. Then there exists irreducible $f_i(x)$ such that $f_i(x)$ divides both $x^m-1$ and $g_{11}(x)f_{12}(x)-g_{12}(x)f_{11}(x)$. Assume that $f_i(x)|g_{11}(x)$. Then $f_i(x)|g_{12}(x)f_{11}(x)$, which implies that $f_i(x)|g_{12}(x)$ or $f_i(x)|f_{11}(x)$, which is not possible. Thus $f_i(x)\nmid g_{11}(x)$. Similarly, $f_i(x)\nmid f_{11}(x)$. Hence, $\gcd\left (\frac{x^m-1}{f_{11}(x)},\frac{x^m-1}{g_{11}(x)},g_{11}(x)f_{12}(x)-g_{12}(x)f_{11}(x)\right )\neq 1$. This completes the proof.   
\end{proof}

\begin{theorem}\label{gen1lcp}
Let $C$ and $D$ be one-generator quasi-cyclic codes generated by $(g_{11}(x),g_{12}(x))$ and $(f_{11}(x),f_{12}(x))$, respectively, where $g_{11}(x),f_{11}(x) \mid (x^m-1)$. Then $(C,D)$ is an LCP of codes if and only if $\gcd (x^m-1, g_{11}(x)f_{12}(x)-g_{12}(x)f_{11}(x))=1$.
\end{theorem}
\begin{proof}
Let $g(x)=\gcd(g_{11}(x),g_{12}(x))$, $g_{11}(x)=g(x)g_{11}'(x)$, $f(x)=\gcd(f_{11}(x),f_{12}(x))$, and $f_{11}(x)=f(x)f_{11}'(x)$. 
 Then  by Lemma \ref{lemma1gen}, we can assume that $C$ is generated by two elements $(g_{11}(x),g_{12}(x))$ and $(0,g_{22}(x))$, and $D$ is generated by two elements $(f_{11}(x),f_{12}(x))$ and $(0,f_{22}(x))$, where 
  \begin{equation*}
      g_{22}(x)=\frac{x^m-1}{g_{11}'(x)} \text{ and }
f_{22}(x)=\frac{x^m-1}{f_{11}'(x)}.
\end{equation*} It is easy to see the generating elements for $C$ and $D$ satisfy Condition $(\ast)$ and $x^m-1=\text{lcm}(g_{11}(x),g_{22}(x))=\text{lcm}(f_{11}(x),f_{22}(x))$ (see Eq. \ref{1geneq}). By Theorem \ref{qslcp}, $(C,D)$ is an LCP of codes if and only if 
\begin{enumerate}
       \item[\textit{(i)}] $\gcd(f_{11}(x),g_{11}(x))=1$.
		\item[\textit{(ii)}] $g(x)=1$ and $f(x)=1$.
		\item[\textit{(iii)}] $\gcd\left (\frac{x^m-1}{g_{11}(x)},\frac{x^m-1}{f_{11}(x)}, g_{11}(x)f_{12}(x)-g_{12}(x)f_{11}(x)\right)=1$.
	\end{enumerate}
We show that $(i),(ii),(iii)$ are equivalent to \begin{align}\label{onelcp}
		\gcd (x^m-1, g_{11}(x)f_{12}(x)-g_{12}(x)f_{11}(x))=1.
	\end{align}
 First, assume that $(i)$, $(ii)$, and $(iii)$ hold. Then,  by Lemma \ref{gen1lemma}, $\gcd (x^m-1, g_{11}(x)f_{12}(x)-g_{12}(x)f_{11}(x))=1$ holds. 

Conversely, if $\gcd(g_{11}(x),f_{11}(x))\neq 1$ or $g(x)\neq 1$ or $f(x)\neq 1$, then $\gcd(x^m-1,g_{11}(x)f_{12}(x)-g_{12}(x)f_{11}(x))\neq 1$. Hence, Eq. \ref{onelcp} implies $(i)$, and $(ii)$. Moreover, by Lemma \ref{gen1lemma}, Eq. \ref{onelcp} along with  $(i)$ and $(ii)$, implies $(iii)$.  This completes the proof.
% This completes the proof.
\end{proof}

\begin{remark}
	The necessary and sufficient condition in Theorem \ref{gen1lcp}, $\gcd(x^m-1,g_{11}(x)f_{12}(x)-g_{12}(x)f_{11}(x))= 1$, can be viewed as an analogue of the vector space case. Let $C$ and $D$ be one-generator quasi-cyclic codes generated by $(g_{11}(x),g_{12}(x))$ and $(f_{11}(x),f_{22}(x))$, respectively. Then $(C,D)$ is an LCP of codes if and only if the $2\times 2$ matrix 
		\begin{align*}
			\begin{bmatrix}
				g_{11}(x) & g_{12}(x)\\
				f_{11}(x) & f_{12}(x)				
			\end{bmatrix}
		\end{align*}
		is invertible.
\end{remark}

% In \cite{Aliabadi2023}, the authors called a code $C$ maximal if $\dim (C)=m$. Recall that if $C$ is a one-generator quasi-cyclic code of length $2m$ and index $2$, generated by $(g_{11}(x),g_{12}(x))$ with $g_{11}(x)\mid (x^m-1)$, then $\dim(C)=m-\deg \ g(x)$, where $g(x)=\gcd(g_{11}(x),g_{12}(x))$.
% \begin{corollary}\cite[Proposition 3.6]{Aliabadi2023}
% Let $C=\langle (g_{11}(x),g_{12}(x))\rangle $ and $D=\langle (f_{11}(x),f_{12}(x))\rangle$ be maximal one-generator quasi-cyclic codes of index $2$, where $g_{11}(x),f_{11}(x) \mid (x^m-1)$. Then $(C,D)$ is an LCP of codes if and only if 
% $\gcd (x^m-1, g_{11}(x)f_{12}(x)-g_{12}(x)f_{11}(x))=1$.
% \end{corollary}

% \begin{remark}
%    Recall that if $C$ is one generator quasi-cyclic code of length $2m$ and index $2$, generated by $(g_{11}(x),g_{12}(x))$ with $g_{11}(x)\mid (x^m-1)$. Then $\dim(C)=m-\deg \ g(x)$, where $g(x)=\gcd(g_{11}(x),g_{12}(x)$. In \cite[Proposition 3.6]{Aliabadi2023}, the authors called code $C$ maximal if $\dim (C)=m$, that is, $g(x)=1$; and they described one generator maximal quasi-cyclic LCP codes of index $2$.
% \end{remark} 
Let $R=\mathbb{F}_q[x]/\langle x^m-1\rangle$. Double circulant (DC) codes are one-generator quasi-cyclic codes of index $2$,   generated by $\langle (1,a(x))\rangle\in R^2$. In the next result, we show that Theorem \ref{gen1lcp} generalizes \cite[Proposition 3.2]{carlet2018}. 
\begin{corollary}
  Let $C=\langle (1,a(x))\rangle$ and $D=\langle (1,b(x))\rangle$  be DC codes in $R^2$. Then $(C,D)$ is an LCP of codes if and only
 if $\gcd(a(x)-b(x),x^m-1)=1$.
\end{corollary}
\begin{proof}
    Take $g_{11}(x)=1, g_{12}(x)=a(x), f_{11}(x)=1, f_{12}(x)=b(x)$. Then, the proof follows from Theorem \ref{gen1lcp}.
\end{proof}

\section{LCPs of quasi-twisted codes of index 2}\label{secqtlcp}
In this section, we characterize the LCP of quasi-twisted codes of index $2$. First, we describe the structure of quasi-twisted codes of index $2$ analogous to that of quasi-cyclic codes. 
\begin{theorem} \label{qt2}
Let $C$ be a $\lambda$-quasi-twisted code of length $2m$ and index $2.$ 
Then $C$ is generated by two elements $(g_{11}(x), g_{12}(x))$ and $(0, g_{22}(x))$ such that they satisfy the following conditions: 
\begin{gather}  
g_{11}(x) \mid (x^m-\lambda) \text{ and } g_{22}(x) \mid(x^m-\lambda), \notag \\
\deg g_{12}(x) < \deg g_{22}(x), \tag{$\ast\ast$}\\
g_{11}(x)g_{22}(x) \mid (x^m-\lambda)g_{12}(x). \notag
\end{gather}
 Moreover, in this case $\dim (C) = 2m-\deg g_{11}(x)-\deg g_{22}(x)$.
\end{theorem}
\begin{proof}
In \cite{lally2001} quasi-cyclic codes (i.e., quasi-twisted codes with $\lambda=1$) were studied using Gr\"{o}bner bases of modules. 
We note that the results of \cite{lally2001}  are valid for any $\lambda$ (all proofs are valid if we take $x^m-\lambda$ instead of 
$x^m-1$). 
Hence, by  \cite{lally2001}   we can assume that $C$ is generated by two elements $g_1=\big( g_{11}(x),g_{12}(x)\big)$ and 
$g_2=\big( 0,g_{22}(x)\big)$, satisfying the conditions 
$g_{11}(x)\mid (x^m-\lambda)$, $g_{22}(x)\mid (x^m-\lambda)$, $\deg g_{12} (x) < \deg g_{22}(x)$, 
and $\dim (C) = 2m- \deg g_{11}(x) - \deg g_{22}(x)$. 
The existence of such generators is equivalent \cite{lally2001} to the existence of a $2\times 2$ polynomial matrix 
$(a_{ij})$ such that 

$$
\left(
\begin{array}{cc}
 a_{11} &  a_{12}  \\
 a_{21} &  a_{22} 
\end{array}
\right)
\left(
\begin{array}{cc}
 g_{11} &  g_{12}  \\
 0 &  g_{22} 
\end{array}
\right) = 
\left(
\begin{array}{cc}
 x^m-\lambda&  0  \\
 0 &  x^m-\lambda
\end{array}
\right).
$$
Then $a_{21}=0$,  $a_{11}=\frac{x^m-\lambda}{g_{11}(x)}$,  $a_{22}=\frac{x^m-\lambda}{g_{22}(x)}$,  and 
$\frac{x^m-\lambda}{g_{11}(x)} g_{12} + a_{12} g_{22} =0$, which implies  $g_{11}(x)g_{22}(x) \mid  (x^m-\lambda) g_{12}(x)$.     
\end{proof}

\begin{remark} \label{remarkqtgcd}   If $\gcd(q,m)=1$, then the condition 
$$g_{11}(x)g_{22}(x) \mid (x^m-\lambda)g_{12}(x)$$ 
in Theorem \ref{qt2} is equivalent to the condition $\gcd(g_{11}(x),g_{22}(x)) \mid g_{12}(x)$, since $x^m-\lambda$ has no multiple roots.
\end{remark}

Recall that a $\lambda$-quasi-twisted code of length $\ell m$ and index $\ell$ is an $R_\lambda$-submodule of $R_\lambda^\ell$. Let $x^m-\lambda=f_1(x)f_2(x)\dots f_s(x)$ be its irreducible factorization over $\mathbb{F}_q$.  Then, we can decompose quasi-twisted codes into their constituent codes analogous to the quasi-cyclic case (for details, see \cite{jia2012}).  A  quasi-twisted code $C$ of index $\ell$ can be  decomposed as 
\begin{equation*}
C \cong \bigoplus_{i=1}^{s} C_i. 
\end{equation*}

In the next theorem, we state a characterization of an LCP of quasi-twisted codes. The proof is similar to that for an LCP of quasi-cyclic codes given in \cite{carlet2018}. We omit the proof. 

\begin{theorem} \label{qtclcp} Let $C$ and $D$ be  $q$-ary $\lambda$-QT codes of length $m\ell$ and index $\ell$ with  CRT decompositions  $C \cong \bigoplus_{i=1}^{s} C_i$  and $D \cong \bigoplus_{i=1}^{s} D_i$, respectively.
Then $(C, D)$ is an LCP of quasi-twisted codes  if and only if 
$(C_i, D_i)$ is an LCP of codes for all $1\leq i\leq s$. 
 \end{theorem}
%  \begin{proof}
% As $C$ and $D$ are $R_\lambda$-submodule of $R_\lambda^\ell$, $(C,D)$ is LCP implies that 
% $$C\oplus D=R_\lambda^\ell.$$
% Consequently,
% \begin{equation}\label{thqt}
%     \left (\bigoplus_{i=1}^{s} C_i \right )\oplus \left (\bigoplus_{i=1}^{s} D_i\right )=\bigoplus_{i=1}^sF_i^\ell.
%     \end{equation}
% That is,    $$\bigoplus_{i=1}^s(C_i+D_i)=\bigoplus_{i=1}^sF_i^\ell.$$ This implies that $C_i+D_i=F_i^\ell$. Therefore, $\ell=\dim(C_i+D_i)=\dim(C_i)+\dim(D_i)-\dim(C_i\cap D_i)$.
%  \end{proof}

% Let $C$ and $D$ be two quasi-twisted codes of length $2m$ and index $2$, generated by the rows of the matrices $G$ and $H$, respectively, where 
% \begin{equation*}
%     G= \begin{bmatrix}
%     g_{11}(x)&g_{12}(x)\\
%     0& g_{22}(x)
% \end{bmatrix} \text{ and } H= \begin{bmatrix}
%     f_{11}(x)&f_{12}(x)\\
%     0& f_{22}(x)
% \end{bmatrix}
% \end{equation*}
% and $f_{ij}(x)$, $g_{ij}(x)$ satisfying Theorem \ref{qt2}.

% Then $C_i$ and $D_i$ are generated by the rows of the matrices
% $$
% G_i=  \begin{bmatrix}
% g_{11}(\eta^{v_i}) & g_{12}(\eta^{v_i}) \\
% 0 & g_{22}(\eta^{v_i}) 
% \end{bmatrix} \text{ and }
% H_i=  \begin{bmatrix}
% f_{11}(\eta^{u_i}) & f_{12}(\eta^{v_i}) \\
% 0 & f_{22}(\eta^{v_i}) 
% \end{bmatrix},$$
% respectively. 
Let $g(x)=\gcd(g_{11}(x),g_{22}(x))$ and $f(x)=\gcd(f_{11}(x),f_{22}(x))$. Then, by Remark \ref{remarkqtgcd}, $g(x)\mid g_{12}(x)$ and $f(x)\mid f_{12}(x)$. Let $g_{11}(x)=g(x)   g_{11}'(x), g_{22}(x)=g(x) g_{22}'(x)$, and $f_{11}(x)=f(x)   f_{11}'(x), f_{22}(x)=f(x) f_{22}'(x)$.
 Next, we provide another characterization of two-generator and one-generator LCPs of quasi-twisted codes of index $2$. 

\begin{theorem}\label{qtslcp}
Let $C$ and $D$ be  $\lambda$-quasi-twisted codes of length $2m$ and index $2$ generated by $(g_{11}(x),g_{12}(x))$, $(0,g_{22}(x))$ and $(f_{11}(x),f_{12}(x))$, $(0,f_{22}(x))$, respectively, satisfying Condition ($\ast\ast$). Then $(C,D)$ is an LCP of quasi-twisted codes if and only if 
	\begin{enumerate}[(I)]
       \item  $\gcd(f_{11}(x),g_{11}(x))=1$.
		\item  $g_{11}(x)g_{22}(x)f_{11}(x)f_{22}(x)=(x^m-\lambda)^2$.
		\item $\gcd(g_{22}'(x),f_{22}'(x), g_{11}(x)f_{12}(x)-g_{12}(x)f_{11}(x))=1$.
	\end{enumerate}
\end{theorem}
\begin{proof}
    The proof is similar to that for  Theorem \ref{qslcp}.
\end{proof}

\begin{theorem}\label{gen1qtlcp}
Let $C$ and $D$ be one-generator $\lambda$-quasi-twisted codes generated by $(g_{11}(x),g_{12}(x))$, and $(f_{11}(x),f_{12}(x))$ respectively, where $g_{11}(x),f_{11}(x) \mid (x^m-\lambda)$. Then $(C,D)$ is an LCP of codes if and only if $\gcd (x^m-\lambda, g_{11}(x)f_{12}(x)-g_{12}(x)f_{11}(x))=1$.
\end{theorem}

\begin{proof}
    The proof is similar to that for Theorem \ref{gen1lcp}.
\end{proof}

\section{Examples of LCPs of codes  with good security parameters}\label{secexample}
 The parameters of a linear code are denoted by $[n,k,d]$, where $n$ is the length, $k$ is the dimension, and $d$ is the minimum Hamming distance of the code. For given $n,k $, we denote the security parameter of an LCP of codes $(C,D)$ by  $d_{LCP}(C,D)=\min\{d(C),d(D^{\perp})\}$, the optimal distance of linear codes by $d_{opt}$, that is, any $[n,k, >d_{opt}]$ code does not exist, and $d_{BKLC}$ denotes the minimum distance of the best known linear code. For the values of $d_{opt}$ and $d_{BKLC}$, we refer to Grassl's code table \cite{codetable}. 
 In \cite{carlet2019}, Carlet et al. proved that for $q>2$, there exists an LCP of linear codes with $d_{LCP}(C,D)=d_{opt}$, and for $q=2$, LCPs of linear codes exist with $d_{LCP}(C,D) \ge d_{opt}-1$.  In this section, we illustrate that the algebraic characterizations in the preceding sections can be utilized to construct LCPs of quasi-cyclic and quasi-twisted codes with $d_{LCP}(C,D)=d_{opt}$ ($q>2$), and $d_{LCP}(C,D)=d_{opt}-1$ or $d_{LCP}(C,D)=d_{opt}\  (q=2)$. 
 % by choosing appropriate generating polynomials.
 All computations are done in MAGMA \cite{magma}.

\begin{example}
 Let $q=5$ and $m=8$. Then $x^m-1=(x+1)(x+2)(x+3)(x+4)(x^2+2)(x^2+3)$  in $\mathbb{F}_5[x]$. Suppose $C$ is generated by
 \begin{equation*}
     \begin{split}
       g_{11}(x)&=(x+2),\\ 
       g_{12}(x)&=(x+2)((x^3 + x^2 + 1)),\\
     g_{22}(x)&=(x+1)(x+2)(x+4)(x^2+2)(x^2+3), 
     \end{split}
 \end{equation*}
  and $D$ is generated by 
  \begin{equation*}
     \begin{split}
       f_{11}(x)&=(x+3),\\ 
       f_{12}(x)&=(x+3)(4x^3 + x^2 + 1),\\
     f_{22}(x)&=(x+1)(x+3)(x+4)(x^2+2)(x^2+3).
     \end{split}
 \end{equation*}
 Then, observe that
 \begin{enumerate}
     \item $\gcd(g_{11}(x),f_{11}(x))=1$.
     \item $g_{11}(x)g_{22}(x)f_{11}(x)f_{22}(x)=(x^8-1)^2$.
     \item $\gcd(g_{22}'(x),f_{22}'(x),g_{11}(x)f_{12}(x)-g_{12}(x)f_{11}(x))=1$.     
 \end{enumerate}
 Thus, by Theorem \ref{qslcp}, $(C,D)$ is an LCP of quasi-cyclic codes with $n=16$, $k=8$, and $d_{LCP}(C,D)=7=d_{opt}$. Moreover, a cyclic code with these parameters does not exist.
\end{example}

\begin{example}
    Let $q=3$ and $m=4$. Suppose $C$ is generated by
 \begin{equation*}
     \begin{split}
       g_{11}(x)&=x+1,\\ 
       g_{12}(x)&=x^2+x+1,\\
        g_{22}(x)&=(x+2)(x^2+1),
     \end{split}
 \end{equation*}
  and $D$ is generated by 
  \begin{equation*}
     \begin{split}
       f_{11}(x)&=x+2,\\ 
       f_{12}(x)&=x^2+x+2,\\
     f_{22}(x)&=(x+1)(x^2+1).
     \end{split}
 \end{equation*}
 Then, by Theorem \ref{qslcp}, $(C,D)$ is an LCP of ternary quasi-cyclic codes with $n=8$, $k=4$, and $d_{LCP}(C,D)=4=d_{opt}$. 
\end{example}

\begin{example}
    Let $q=3$ and $m=5$. Suppose $C$ is generated by
 \begin{equation*}
     \begin{split}
       g_{11}(x)&=x+2,\\ 
       g_{12}(x)&=x^2+2x+2,\\
     g_{22}(x)&=x^4+x^3+x^2+x+1,
     \end{split}
 \end{equation*}
  and $D$ is generated by 
  \begin{equation*}
     \begin{split}
       f_{11}(x)&=1,\\ 
       f_{12}(x)&=x^3 + 2x^2 + 2x + 1,\\
       f_{22}(x)&=x^5-1.
     \end{split}
 \end{equation*}
 Then $(C,D)$ is an LCP of ternary quasi-cyclic codes with $n=10$, $k=5$ and $d_{LCP}(C,D)=5=d_{opt}$. This security parameter improves the security parameter obtained in  \cite[Table 3]{Aliabadi2023} with $n=10$ and $k=5$.
\end{example}
In Table \ref{tab:my_label}, we list some examples of LCPs of quasi-cyclic codes $(C,D)$ over $\mathbb{F}_4=\mathbb{F}_2[w]$, where the code $C$ is generated by $(g_{11}(x),g_{12}(x)), (0,g_{22}(x))$ and the code $D$ is generated by $(f_{11}(x),f_{12}(x)), (0,f_{22}(x))$. We compare the security parameter $d_{LCP}(C,D)$ with $d_{BKLC}$.

\begin{table}[h!]
    \centering
\small\begin{tabular}{|c|c|c|c|c|c|c|c|c|c|}
    \hline
     $n$&$k$  & $g_{11}(x)$& $g_{12}(x)$& $g_{22}(x)$& $f_{11}(x)$& $f_{12}(x)$& $f_{22}(x)$ & $d_{LCP}$ & $d_{opt}$  \\
    %    \hline
    % $6$ & $3$ & $x+w$ & $x+1 + w^2$ &\makecell{$x^2 + wx $\\$+ w^2$}& $1$&  \makecell{$wx^3+wx^2 $\\$+ w^2x + 1$} & $x^3-1$ & $3$ &$4$ \\
       \hline
       $10$ & $5$ & $x+1$ &\makecell{$x^3 + wx^2 $\\$+ wx$} &\makecell{$x^4 + x^3 $\\$+ x^2 + x $\\$+ 1$}& $1$& \makecell{$x^4 + wx^3 $\\$+ x + w$} & $x^5-1$ & $5$ &$5$ \\
       \hline
       $14$ & $9$ & $x+1$& \makecell{$x^3 + w^2x$\\ $+ w$} &\makecell{$x^4 + x^3$ \\ $+ x^2 + 1$}& \makecell{$x^3+x^2$\\$+1$}& \makecell{$x^3+x^2$\\$+1$ }& \makecell{$x^6 + x^5 $\\ $+ x^4 + x^3 $\\$+ x^2 + x + 1$} & $4$ &$4$ \\
       \hline
       % $18$& $9$ & \makecell{$x^2 + w^2x$\\$ + w$} & \makecell{$x^3 + w^2x^2 $\\$+ x + 1$} &\makecell{$x^7 + wx^6 $\\$+ x^4 + wx^3$\\$ + x + w$}& $x+w^2$& \makecell{$wx^4 + x^3$ \\$+ x^2 + x$\\$ + 1$} & \makecell{$x^8 + wx^7 $\\$+ w^2x^6 + x^5 $\\$+ wx^4 + w^2x^3 $\\$+ x^2 + wx + w^2$} & $6$ &$8$ \\
       % \hline
                
    \end{tabular}

    \caption{LCPs of quasi-cyclic codes over $\mathbb{F}_4$.}
    \label{tab:my_label}
\end{table}
Next, we provide some examples of LCPs of quasi-twisted codes with good security parameters. In the following example, we get the optimal security parameter over $\mathbb{F}_5$ with $n=22$ and $k=11$.
\begin{example}
    Let $q=5,$ $m=11$ and $\lambda=2$. Suppose $C$ is a $\lambda$-quasi-twisted code generated by 
     \begin{equation*}
     \begin{split}
       g_{11}(x)&=x+2,\\ 
       g_{12}(x)&=x^6 + 4x^5 + 2x^3 + 3x^2 + x + 3,\\
     g_{22}(x)&=(x^5 + x^4 + x^3 + 2x^2 + x + 2)(x^5 + 2x^4 + x^3 + 2x^2 + 3x + 2),
     \end{split}
 \end{equation*}
  and $D$ is an another $\lambda$-quasi-twisted code generated by 
  \begin{equation*}
     \begin{split}
       f_{11}(x)&=1,\\ 
       f_{12}(x)&=x^9 + 2x^8 + x^7 + 2x^6 + 4x^5 + 2x^4 + x^2 + 2x + 1,\\
       f_{22}(x)&=x^{11}-2.
     \end{split}
 \end{equation*}
 Then, by Theorem \ref{qtslcp}, $(C,D)$ is an LCP of $\lambda$-quasi-twisted quinary codes with $n=22$, $k=11$, and $d_{LCP}(C,D)=8=d_{BKLC}$. 
\end{example}

In \cite{Bouyuklieva2020}, the author showed that $d_{LCD}(C)\leq 7$ for $n=30$ and $k=15$ over $\mathbb{F}_2$, where $d_{LCD}(C)$ is the minimum distance of a binary LCD code $C$. In the following example, we show that there is an LCP of codes with $d_{LCP}(C,D)=7$ for the same $n$ and $k$.
%\textcolor{blue}{ Should we write it or not?} 

\begin{example}
	Let $q=2,$ $m=15$. Suppose $C$ is a quasi-cyclic code generated by 
	\begin{equation*}
		\begin{split}
			g_{11}(x)&=x^2 + x + 1,\\ 
			g_{12}(x)&=x^{12} + x^{11}+ x^9 + x^8 + x^7 + x^5+ x^4 + x^2,\\
			g_{22}(x)&=(x^4 + x^3 + 1)(x^4 + x^3 + x^2 + x + 1)(x^4 + x + 1)(x+1),
		\end{split}
	\end{equation*}
	and $D$ is an another quasi-cyclic code generated by 
	\begin{equation*}
		\begin{split}
			f_{11}(x)&=x + 1,\\ 
			f_{12}(x)&=x^{12} + x^{10}+ x^8 + x^7 + x^6+ x + 1,\\
			f_{22}(x)&=(x^4 + x^3 + x^2 + x + 1)(x^4 + x + 1)(x^4 + x^3 + 1)(x^2 + x + 1).
		\end{split}
	\end{equation*}
	Then, by Theorem \ref{qslcp}, $(C,D)$ is a binary LCP of quasi-cyclic codes with $d_{LCP}(C,D)=7=d_{opt}-1$, where $C$ is a binary $[30,15]$ code.
\end{example}

\begin{example}
	Let $q=2,$ $m=15$. Suppose $C$ is a quasi-cyclic code generated by 
	\begin{equation*}
		\begin{split}
			g_{11}(x)&=x+1,\\ 
			g_{12}(x)&=(x+1)(x^6+x^5+x^2+1),\\
			g_{22}(x)&=x^m-1,
		\end{split}
	\end{equation*}
	and $D$ is another quasi-cyclic code generated by 
	\begin{equation*}
		\begin{split}
			f_{11}(x)&= 1,\\ 
			f_{12}(x)&=x^7+x^5+x^4+x^3+1,\\
			f_{22}(x)&=(x^4 + x^3 + x^2 + x + 1)(x^4 + x + 1)(x^4 + x^3 + 1)(x^2 + x + 1).
		\end{split}
	\end{equation*}
	Then, by Theorem \ref{qslcp}, $(C,D)$ is a binary LCP of quasi-cyclic codes with $d_{LCP}(C,D)=8=d_{opt}$, where $C$ is a binary $[30,14]$ code. This value of $d_{LCP}(C,D)$ is better than the lower bound $d_{opt}-1$. 
\end{example}

\begin{example}
	Let $q=2,$ $m=21$. Suppose $C$ is a quasi-cyclic code generated by 
	\begin{equation*}
		\begin{split}
			g_{11}(x)&=(x+1)(x^3 + x^2 + 1)(x^6 + x^4 + x^2 + x + 1),\\ 
			g_{12}(x)&= (x + 1)^3(x^3 + x^2 + 1)^2(x^5 + x^2 + 1)(x^6 + x^4 + x^2 + x + 1),\\
			g_{22}(x)&=x^m-1,
		\end{split}
	\end{equation*}
and $D$ is another quasi-cyclic code generated by 
	\begin{equation*}
		\begin{split}
			f_{11}(x)&= 1,\\ 
			f_{12}(x)&= x^9 + x^6 + x^3 + x + 1,\\
			f_{22}(x)&= (x^2 + x + 1)(x^3 + x + 1)(x^6 + x^5 + x^4 + x^2 + 1).
		\end{split}
	\end{equation*}
	Then $(C,D)$ is a binary LCP of quasi-cyclic codes with $d_{LCP}(C,D)=16=d_{opt}$, where $C$ is a binary $[42,11,16]$ code. This value of $d_{LCP}(C,D)$ is better than the lower bound $d_{opt}-1$. 
\end{example}

\section*{Conclusion}
Linear complementary pairs (LCP) of codes generalize linear complementary dual (LCD) codes and have applications in resistance against side-channel and fault-injection attacks. In this work, we presented a polynomial characterization of LCPs of quasi-cyclic and quasi-twisted codes of index $2$. As a consequence, we provided necessary and sufficient conditions for LCPs of one-generator quasi-cyclic and quasi-twisted codes. 
In addition, we provided several examples of LCPs of quasi-cyclic and quasi-twisted codes with optimal security parameters. 
%\textcolor{blue}{An interesting direction for future research is to explicitly derive the security parameters of LCPs of quasi-cyclic codes in terms of their generating polynomials. Furthermore, extending this framework to obtain a complete polynomial characterization of linear complementary pairs of quasi-cyclic (quasi-twisted) codes for arbitrary index presents a significant theoretical challenge. The approach developed in this paper could, in principle, be adapted to this setting; however,  it would require handling a larger system of polynomials concurrently, and careful control of their algebraic relationships to guarantee complementary properties. The combinatorial and algebraic complexities grow rapidly with the index, making the necessary polynomial congruences, gcd conditions, and complementary constraints much more complicated.}
 
 %while polynomial characterizations exist for low indices (especially index 2), extending these methods to arbitrary index

\section*{Acknowledgment}

K. Abdukhalikov and G. K. Verma were supported by UAEU grants G00004233 and G00004614. 
D. Ho was supported by   the Tromsø Research Foundation (project “Pure Mathematics in
Norway”), and  UiT Aurora project MASCOT.
S. Ling is supported by Nanyang Technological University Research Grant No. 04INS000047C230GRT01.

\bibliographystyle{abbrv}
	\bibliography{ref}

\begin{thebibliography}{10}

\bibitem{Abdukhalikov2023}
K.~Abdukhalikov, T.~Bag, and D.~Panario.
\newblock One-generator quasi-cyclic codes and their dual codes.
\newblock {\em Discret. Math.}, 346(6):113369, 2023.

\bibitem{Abdukhalikov20251}
K.~Abdukhalikov, A.~S. Dzhumadil’daev, and S.~Ling.
\newblock Quasi-cyclic codes of index 2.
\newblock {\em Discrete Math.}, 349(6):Paper No. 115004, 2026.

\bibitem{Abdukhalikov20252}
K.~Abdukhalikov, D.~Ho, S.~Ling, and G.~K. Verma.
\newblock Linear complementary dual quasi-cyclic codes of index 2.
\newblock {\em ArXiv.2504.09126}.

\bibitem{ackerman2011}
R.~Ackerman and N.~Aydin.
\newblock New quinary linear codes from quasi-twisted codes and their duals.
\newblock {\em Appl. {M}ath. {L}ett}, 24(4):512--515, 2011.

\bibitem{Aliabadi2023}
Z.~Aliabadi, C.~G{\"u}neri, and T.~Kalaycı.
\newblock On the hull and complementarity of one generator quasi-cyclic codes
  and four-circulant codes.
\newblock {\em J. Algebra its Appl.}, doi: 10.1142/S0219498825410166, 2025.

\bibitem{aydin2017}
N.~Aydin, N.~Connolly, and M.~Grassl.
\newblock Some results on the structure of constacyclic codes and new linear
  codes over {GF(7)} from quasi-twisted codes.
\newblock {\em Adv. Math. Commun.}, 11(1):245--258, 2017.

\bibitem{magma}
W.~Bosma, J.~Cannon, and C.~Playoust.
\newblock The {M}agma algebra system. {I}. {T}he user language.
\newblock {\em J. Symbolic Comput.}, 24(3-4):235--265, 1997.
\newblock Computational algebra and number theory (London, 1993).

\bibitem{Bouyuklieva2020}
S.~Bouyuklieva.
\newblock Optimal binary {LCD} codes.
\newblock {\em Designs, Codes and Cryptography}, 89:2445 -- 2461, 2021.

\bibitem{Carlet2016}
C.~Carlet and S.~Guilley.
\newblock Complementary dual codes for counter-measures to side-channel
  attacks.
\newblock {\em Adv. Math. Commun.}, 10:131--150, 2016.

\bibitem{carlet2018}
C.~Carlet, C.~G\"uneri, F.~\"Ozbudak, B.~\"Ozkaya, and P.~Sol\'e.
\newblock On linear complementary pairs of codes.
\newblock {\em IEEE Trans. Inform. Theory}, 64(10):6583--6589, 2018.

\bibitem{carlet2019}
C.~Carlet, S.~Mesnager, C.~Tang, and Y.~Qi.
\newblock On $\sigma$-{LCD} codes.
\newblock {\em IEEE Transactions on Information Theory}, 65(3):1694--1704,
  2019.

\bibitem{chepyzhov1993}
V.~Chepyzhov.
\newblock A {G}ilbert-{V}ashamov bound for quasi-twisted codes of rate 1/n.
\newblock In {\em Proc. Joint Swedish-Russian Int. Workshop on Inf. Theory},
  pages 214--218, 1993.

\bibitem{Choi2022}
W.-H. Choi, C.~G{\"u}neri, J.-L. Kim, and F.~{\"O}zbudak.
\newblock Optimal binary linear complementary pairs of codes.
\newblock {\em Cryptography and Communications}, 15:469--486, 2022.

\bibitem{codetable}
M.~Grassl.
\newblock Bounds on the minimum distance of linear codes and quantum codes.
\newblock {\em Online available at: http://s208785153.online.de/codetables/}.

\bibitem{Guneri2023}
C.~G{\"u}neri.
\newblock Optimal binary linear complementary pairs from solomon-stiffler
  codes.
\newblock {\em IEEE Transactions on Information Theory}, 69:6512--6517, 2023.

\bibitem{Guneri2018}
C.~G{\"u}neri, B.~{\"O}zkaya, and S.~Sayıcı.
\newblock On linear complementary pair of $n$d cyclic codes.
\newblock {\em IEEE Commun. Lett.}, 22:2404--2406, 2018.

\bibitem{Huffman2021}
W.~C. Huffman, J.-L. Kim, and P.~Sol{\'e}.
\newblock Concise encyclopedia of coding theory.
\newblock {\em Chapman and Hall/CRC}, 2021.

\bibitem{jia2012}
Y.~Jia.
\newblock On quasi-twisted codes over finite fields.
\newblock {\em Finite Fields Appl.}, 18(2):237--257, 2012.

\bibitem{kasami1974}
T.~Kasami.
\newblock A {G}ilbert-{V}arshamov bound for quasi-cycle codes of rate $1/2$
  (corresp.).
\newblock {\em IEEE Trans. Inform. Theory}, 20(5):679--679, 1974.

\bibitem{lally2001}
K.~Lally and P.~Fitzpatrick.
\newblock Algebraic structure of quasicyclic codes.
\newblock {\em Discrete Appl. Math.}, 111(1-2):157--175, 2001.

\bibitem{Li2023}
S.~Li, M.~Shi, and S.~Ling.
\newblock An open problem and a conjecture on binary linear complementary pairs
  of codes.
\newblock {\em IEEE Transactions on Information Theory}, 71:219--226, 2023.

\bibitem{Ling2006}
S.~Ling, H.~Niederreiter, and P.~Sol{\'e}.
\newblock On the algebraic structure of quasi-cyclic codes {IV}: Repeated
  roots.
\newblock {\em Des. Codes, Cryptogr.}, 38:337--361, 2006.

\bibitem{Ling2001}
S.~Ling and P.~Sol{\'e}.
\newblock On the algebraic structure of quasi-cyclic codes {I}: Finite fields.
\newblock {\em IEEE Trans. Inf. Theory}, 47:2751--2760, 2001.

\bibitem{Ling2003}
S.~Ling and P.~Sol{\'e}.
\newblock On the algebraic structure of quasi-cyclic codes {II}: Chain rings.
\newblock {\em Des. Codes, Cryptogr.}, 30:113--130, 2003.

\bibitem{Ling2005}
S.~Ling and P.~Sol{\'e}.
\newblock On the algebraic structure of quasi-cyclic codes {III}: generator
  theory.
\newblock {\em IEEE Trans. Inform. Theory}, 51:2692--2700, 2005.

\bibitem{Luo2023}
G.~Luo, M.~F. Ezerman, S.~Ling, and B.~{\"O}zkaya.
\newblock Improved spectral bound for quasi-cyclic codes.
\newblock {\em IEEE Trans. Inform. Theory}, 70:4002--4015, 2023.

\bibitem{Ngo2015}
X.~T. Ngo, S.~Bhasin, J.-L. Danger, S.~Guilley, and Z.~Najm.
\newblock Linear complementary dual code improvement to strengthen encoded
  circuit against hardware trojan horses.
\newblock {\em 2015 IEEE International Symposium on Hardware Oriented Security
  and Trust (HOST)}, pages 82--87, 2015.

\bibitem{qian2019}
L.~Qian, M.~Shi, and P.~Sol{\'e}.
\newblock On self-dual and {LCD} quasi-twisted codes of index two over a
  special chain ring.
\newblock {\em Cryptogr. Commun.}, 11:717--734, 2019.

\bibitem{Semenov2012}
P.~Semenov and P.~Trifonov.
\newblock Spectral method for quasi-cyclic code analysis.
\newblock {\em IEEE Commun. Lett.}, 16:1840--1843, 2012.

\bibitem{Townsend1967}
R.~Townsend and E.~J. Weldon.
\newblock Self-orthogonal quasi-cyclic codes.
\newblock {\em IEEE Trans. Inf. Theory}, 13:183--195, 1967.

\bibitem{wu2020}
R.~Wu and M.~Shi.
\newblock A modified {G}ilbert-{V}arshamov bound for self-dual quasi-twisted
  codes of index four.
\newblock {\em Finite Fields Appl.}, 62:101627, 2020.

\bibitem{Zeh2015}
A.~Zeh and S.~Ling.
\newblock Spectral analysis of quasi-cyclic product codes.
\newblock {\em IEEE Trans. Inform. Theory}, 62:5359--5374, 2015.

\end{thebibliography}
\end{document}